\documentclass[titlerunning]{llncs}

\usepackage{
amsmath,
amscd,
amssymb,
hyperref
}
\usepackage{comment}
\usepackage{tikz}
\usetikzlibrary{positioning,arrows,decorations.pathmorphing,decorations.pathreplacing}
\usetikzlibrary{calc}
\usepackage[skip=-2pt]{caption}

\usepackage{graphicx}

\usepackage{url}
\urldef{\mails}\path|{vosalo, iatorm}@utu.fi|
\newcommand{\keywords}[1]{\par\addvspace\baselineskip
\noindent\keywordname\enspace\ignorespaces#1}

\newcommand{\Z}{\mathbb{Z}}

\newcommand{\N}{\mathbb{N}}

\newcommand{\D}{\mathcal{S}}
\newcommand{\F}{\mathcal{F}}
\newcommand{\Ls}{\mathcal{L}}
\newcommand{\Locc}{\mathcal{L}_{\mathrm{occ}}}
\newcommand{\Ts}{\mathcal{T}}

\begin{document}

\mainmatter

\title{A One-Dimensional Physically Universal Cellular Automaton\thanks{Research supported by the Academy of Finland Grant 131558}}

\author{
Ville Salo\inst{1}\thanks{The author was partially supported by CONICYT Proyecto Anillo ACT 1103}
\and
Ilkka T\"orm\"a\inst{2}
}

\institute{
		Center for Mathematical Modeling, \\
		University of Chile \\
\and
		TUCS -- Turku Centre for Computer Science, \\
		University of Turku, Finland \\
		\mails
}

\maketitle

\begin{abstract}
Physical universality of a cellular automaton was defined by Janzing in 2010 as the ability to implement an arbitrary transformation of spatial patterns. In 2014, Schaeffer gave a construction of a two-dimensional physically universal cellular automaton. We construct a one-dimensional version of the automaton.
\keywords{cellular automaton, physical universality, reversibility}
\end{abstract}

\section{Introduction}

A cellular automaton (CA) is a finite or infinite lattice of deterministic finite state machines with identical interaction rules, which, at discrete time steps, update their states simultaneously based on those of their neighbors. They are an idealized model of massively parallel computation. From another point of view, the local updates can be seen as particle interactions, and the CA is then a kind of physical law, or dynamics, governing the universe of all state configurations.

We study the notion of \emph{physical universality} of cellular automata, introduced by Janzing in \cite{Ja10}, which combines the two viewpoints in a novel and interesting way. Intuitively, a cellular automaton is physically universal if, given a finite subset $D$ of the lattice and a function $h$ on the shape-$D$ patterns over the states of the CA, one can build a `machine' in the universe of the CA that, under the dynamics of the CA, decomposes any given pattern $P$ and replaces it by $h(P)$.

A crucial point in this definition is that we need to perform arbitrary computation on all patterns, not only carefully constructed ones. This has quite serious implications: The machine $M$ that takes apart an arbitrary pattern of shape $D$ and replaces it by its image under an arbitrary function $h$ is not in any way special, and in particular we are not allowed to have separate `machine states' and `data states' with the former operating on the latter. Instead, we can also think of $M$ as a pattern, and must be able to construct a larger machine $M'$ that takes $M$ apart and reassembles it in an arbitrary way.

This notion differs essentially from most existing notions of universality for CA such as \emph{intrinsic universality} \cite{Ol03}, universality in terms of traces as discussed by Delvenne et al. in \cite{DeKuBl06}, and the more well-known I-know-it-when-I-see-it type of computational universality promoted by Wolfram in \cite{Wo02}. In these notions, one can usually implement the computations and simulations in a well-behaved subset of configurations. Physical universality bears more resemblance to the \emph{universal constructor machines} of Von Neumann \cite{Ne66}, which construct copies of themselves under the dynamics of a particular cellular automaton, and were the initial motivation for the definition of CA. Another property of CA with a similar flavor is \emph{universal pattern generation} as discussed in \cite{Ka12}, meaning the property of generating all finite patterns from a given simple initial configuration.

In Janzing's work \cite{Ja10} some results were already proved about physically universal CA, but it was left open whether such an object actually exists. A two-dimensional physically universal CA was constructed by Schaeffer in \cite{Sc14}, but it was left open whether this CA can be made one-dimensional. We construct such a CA, solving the question in the positive. %This is done in sections~\ref{sec:TheCA} through~\ref{sec:AllTogether}. The reader may want to begin by taking a glance at Section~\ref{sec:AllTogether} for a general outline of the construction.

%In Section~\ref{sec:Further}, we outline the proofs of two additional results. Namely, we construct a reversibly physically universal CA and give a nontrivial example of a physically nonuniversal CA. In Section~\ref{sec:Other}, we briefly discuss the definition of physical universality in the context of symbolic dynamics, in particular suggesting a definition of physical universality for cellular automata on mixing SFTs. We also state some open questions.

\section{Definitions}

Let $A$ be a finite set, called the \emph{state set}, and $\Z^d$ a grid. A \emph{cellular automaton} is a map $f : A^{\Z^d} \to A^{\Z^d}$ defined by a \emph{local function} $F : A^N \to A$, where $N \subset \Z^d$ is a finite \emph{neighborhood}, so that $f(x)_{\vec v} = F(x_{N + \vec v})$ for all $x \in A^{\Z^d}$ and $\vec v \in \Z^d$. It is \emph{reversible} if there is another CA $g : A^{\Z^d} \to A^{\Z^d}$ such that $f \circ g = g \circ f = \mathrm{id}$. An example of a reversible CA is the \emph{shift by $\vec n \in \Z^d$}, defined by $\sigma^{\vec n}(x)_{\vec v} = x_{\vec v + \vec n}$. Other examples can be constructed as follows. Let $A$ be a Cartesian product $\prod_{i = 1}^k A_i$ with projection maps $\pi_i : A \to A_i$, let $\vec n_1, \ldots, \vec n_k \in \Z^d$ be arbitrary, and let $\gamma : A \to A$ be a bijection. Then the CA $f$ defined by $f(x)_{\vec v} = \gamma(\pi_1(x_{\vec v + \vec n_1}), \ldots, \pi_k(x_{\vec v + \vec n_k}))$ is reversible. We call $f$ a \emph{partitioned CA}, and the components $A_i$ are \emph{tracks}. In the CA, the tracks are first shifted individually by the $\vec n_i$, and then the bijection $\gamma$ is applied to every cell.

A CA $f$ is \emph{physically universal} if the following condition holds. For all finite domains $D, E \subset \Z^d$, and all functions $h : A^D \to A^E$, there exists a partial configuration $x \in A^{\Z^d \setminus D}$ and a time $t \in \N$ such that for all $P \in A^D$, we have $f^t(x \cup P)_E = h(P)$. It is \emph{effectively physically universal} if $t$ is polynomial in the diameter of $D \cup E$ and the computational complexity of $h$ according to some `reasonable' complexity measure. In this article, we use circuit complexity, or more precisely, the number of binary NAND gates needed to implement $h$. One could reasonably require also that the configuration $x$ is computed efficiently from the efficient presentation of the function $h$. Our proof gives a polynomial time algorithm for this. See Section~\ref{sec:Final} for a discussion.

\section{The Cellular Automaton}
\label{sec:TheCA}

Our physically universal automaton is a partitioned CA $f$ defined as follows.
\begin{itemize}
\item The alphabet is $A = \{0, 1\}^4$.
\item The shifts are 2, 1, $-1$ and $-2$, and we denote $\D = \{2, 1, -1, -2\}$.
\item For each $a, b \in \{0, 1\}$ bijection $\gamma$ maps the state $(1, a, b, 1)$ to $(1, b, a, 1)$, and $(a, 1, 1, b)$ to $(b, 1, 1, a)$. Everything else is mapped to itself.
\end{itemize}
Intuitively, in the CA $f$ there are four kinds of particles: fast right, slow right, slow left and fast left. At most one particle of each kind can be present in a cell. On each step, every slow (fast) particle moves one step (two steps, respectively) in its direction. After that, if two fast or two slow particles are in the same cell, then the direction of every particle of the other speed is reversed. This resembles the two-dimensional CA of Schaeffer, where particles move to four directions (NE, NW, SE, SW) with speed one, and the head-on collision of two particles causes other particles in the same cell to make a u-turn.

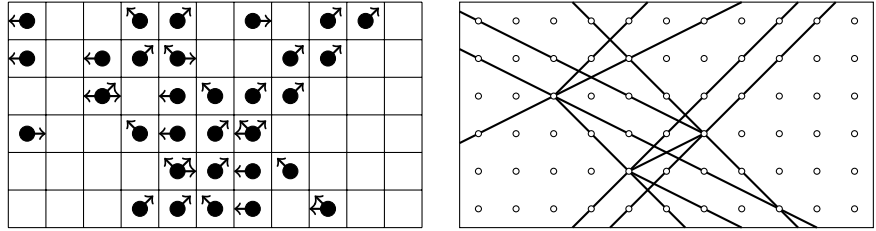
\begin{figure}
\begin{center}
\begin{tikzpicture}[yscale=-1]

\draw[step=.5] (-.5,0) grid (5,3);

%bullets
\foreach \x/\y in {
     	         2/5,3/5,4/5,5/5,    7/5,
                  3/4,4/4,5/4,6/4,
-1/3,         2/3,3/3,4/3,5/3,
          1/2,    3/2,4/2,5/2,6/2,
-1/1,     1/1,2/1,3/1,        6/1,7/1,
-1/0,         2/0,3/0,    5/0,    7/0,8/0}{
	\draw[fill=black] (\x/2+.25,\y/2+.25) circle (.1);
}
%2-arrows
\foreach \x/\y in {
                  3/4,
-1/3,
          1/2,
                  3/1,
                          5/0}{
	\draw[thick,->] (\x/2+.25,\y/2+.25) -- ++(0:.25);
}
%1-arrows
\foreach \x/\y in {
	    2/5,3/5,
             3/4,4/4,
                 4/3,5/3,
     1/2,            5/2,6/2,
         2/1,            6/1,7/1,
             3/0,            7/0,8/0}{
	\draw[thick,->] (\x/2+.25,\y/2+.25) -- ++(-45:.25);
}
%-1-arrows
\foreach \x/\y in {
	            4/5,        7/5,
             3/4,        6/4,
         2/3,        5/3,
                 4/2,
             3/1,
         2/0}{
	\draw[thick,->] (\x/2+.25,\y/2+.25) -- ++(225:.25);
}
%-2-arrows
\foreach \x/\y in {
                          5/5,    7/5,
                          5/4,
                  3/3,    5/3,
          1/2,    3/2,
-1/1,     1/1,
-1/0}{
	\draw[thick,->] (\x/2+.25,\y/2+.25) -- ++(180:.25);
}

\begin{scope}[xshift=6cm]

\draw (-.5,0) rectangle (5,3);

\draw[thick] (1,3) -- (4,0);
\draw[thick] (1.5,3) -- (4.5,0);
\draw[thick] (2.5,3) -- (.75,1.25) -- (2,0);
\draw[thick] (4,3) -- (1,0);
\draw[thick] (3.25,3) -- (1.75,2.25) -- (2.75,1.75) -- (-.5,.125);
\draw[thick] (4.25,3) -- (-.5,.625);
\draw[thick] (-.5,1.875) -- (3.25,0);

\foreach \x in {0,...,10}{
	\foreach \y in {0,...,5}{
		\draw[fill=white] (\x/2-.25,\y/2+.25) circle (.04);
	}
}

\end{scope}

\end{tikzpicture}
\end{center}
\caption{A sample spacetime diagram of $f$, and a schematic version on the right. Time increases upward. Particles are represented by arrowed bullets. For example, a bullet with arrows to the east, northeast and northwest represents three particles moving at speeds 2, 1 and $-1$, respectively.}

\label{fig:Spacetime}
\end{figure}

\section{The Logical Cellular Automaton}

For the proof of physical universality, we define another CA on an infinite alphabet. For this, define the \emph{ternary conditional operator} by $p(a, b, c) = c \oplus (a \wedge (c \oplus b))$ for all $a, b, c \in \{0,1\}$. That is, $p(a, b, c)$ is equal to $b$ if $a = 1$, and to $c$ otherwise. In many programming languages, $p(a, b, c)$ is denoted by $a\;?\;b:c$.

\begin{definition}
Let $\mathcal{V} = \{\alpha_1, \alpha_2, \ldots\}$ be an infinite set of variables, and denote by $\F$ the set of Boolean functions over finitely many variables of $\mathcal{V}$. The \emph{logical extension of $f$} is the CA-like function $\hat f : \hat A^\Z \to \hat A^\Z$ on the infinite alphabet $\hat A = \F^4$, where the four tracks are first shifted as in $f$, and then the function
\[ (a, b, c, d) \mapsto (p(b \wedge c, d, a), p(a \wedge d, c, b), p(a \wedge d, b, c), p(b \wedge c, a, d)) \]
is applied to each coordinate. A \emph{valuation} is a function $v : \mathcal{V} \to \{0, 1\}$. It extends to $\F$ and then into a function $v : \hat A^\Z \to A^\Z$ in the natural way. 
\end{definition}

The logical extension simulates multiple spacetime diagrams of $f$: one can see that the definition of $\hat f$ is equal to that of $f$, except that each particle is replaced by a Boolean formula that corresponds to the conditional presence or absence of a particle. We think of $A$ as a subset of $\hat A$ containing the constant 0 or constant 1 function in each track. Note that $\hat f$ is also reversible, and we denote by $\hat f^{-1}$ its inverse function. See Figure~\ref{fig:SpacetimeLogic} for a spacetime diagram of $\hat f$.

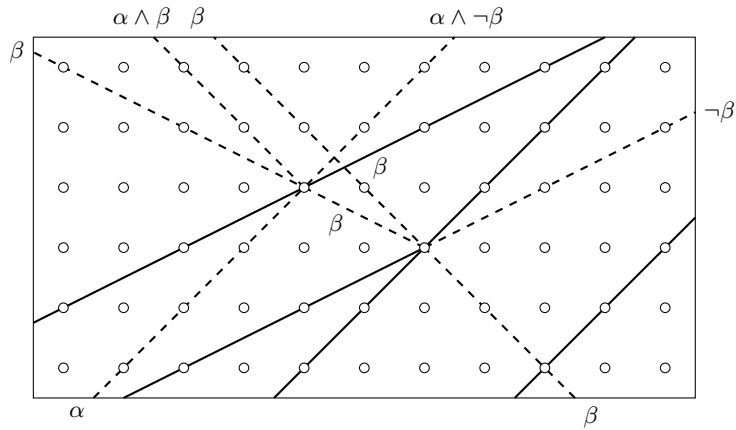
\begin{figure}
\begin{center}
\begin{tikzpicture}[scale=1.6,yscale=-1]

\draw (-.5,0) rectangle (5,3);

\draw[thick,dashed] (0,3) -- (3,0);
\draw[thick,dashed] (1.75,1.25) -- (.5,0);
\draw[thick] (1.5,3) -- (4.5,0);
\draw[thick,dashed] (4,3) -- (1,0);
\draw[thick] (.25,3) -- (2.75,1.75);
\draw[thick,dashed] (2.75,1.75) -- (5,.625);
\draw[thick,dashed] (2.75,1.75) -- (-.5,.125);
\draw[thick] (-.5,2.375) -- (4.25,0);
\draw[thick] (3.5,3) -- (5,1.5);

\foreach \x in {0,...,10}{
	\foreach \y in {0,...,5}{
		\draw[fill=white] (\x/2-.25,\y/2+.25) circle (.04);
	}
}

\node[below left] () at (0,3) {$\alpha$};
\node[below right] () at (4,3) {$\beta$};
\node[right] () at (5,.625) {$\neg \beta$};
\node[above right] () at (2.25,1.25) {$\beta$};
\node[below left] () at (2.15,1.4) {$\beta$};
\node[above left] () at (1,0) {$\beta$};
\node[left] () at (-.5,.125) {$\beta$};
\node[above] () at (.4,0) {$\alpha \wedge \beta$};
\node[above] () at (3.1,0) {$\alpha \wedge \neg \beta$};

\end{tikzpicture}
\end{center}
\caption{A schematic spacetime diagram of $\hat f$, where $\alpha, \beta \in \F$. Particles are represented by solid lines, and Boolean particles by dotted lines. Note that Boolean particles can be created, even though $f$ itself conserves the number of particles.}
\label{fig:SpacetimeLogic}
\end{figure}

The following result holds basically by construction.

\begin{lemma}
\label{lem:LogicRep}
Let $x \in \hat A^\Z$ be a configuration, and let $v : \mathcal{V} \to \{0, 1\}$ be a valuation, so that $v(x) \in A^\Z$. Then for all $t \in \Z$ we have $f^t(v(x)) = v(\hat f^t(x))$.
\end{lemma}

The idea of the proof of physical universality of $f$ using this new CA is the following. We may assume that $D = E = [0, n-1]$ in the definition of physical universality, for some $n \in \N$. Then, we construct a spacetime diagram of $\hat f$ with the following properties. First, in the initial configuration $x \in \hat A^\Z$, the cells of the interval $[0, n-1]$ contain $4n$ distinct variables from $\mathcal{V}$. All other cells of $x$ contain either 0 or 1. There also exists $t > 0$ such that $\hat f^t(x)_{[0, n-1]}$ contains the Boolean functions computing the function $h$ in the definition of physical universality. In the course of the construction, we define which cells of $x$ contain a 1.

\begin{definition}

We introduce the following terminology for the construction.
\begin{itemize}

\item The \emph{configuration of interest}, denoted by $x \in \hat A^\Z$, initially contains the `fully general' state $(\alpha_{4i}, \alpha_{4i+1}, \alpha_{4i+2}, \alpha_{4i+3})$ in every cell $i \in [0, n-1]$, and 0 everywhere else. During the construction, we change the cells of $x$ from 0 to 1, but keep referring to it as $x$, so some of the definitions below depend on the stage of the proof. The \emph{spacetime diagram of interest} is defined similarly.

% perhaps the following would be better, in the name of clarity
%\item By a \emph{configuration of interest} we mean any point $x \in \hat A^\Z$, that contains the ``fully general'' state $(\alpha_{4i}, \alpha_{4i+1}, \alpha_{4i+2}, \alpha_{4i+3})$ in every cell $i \in [0, n-1]$.

\item A \emph{(spacetime) position} is an element of $\Z \times \D$ ($\Z \times \Z \times \D$, respectively), representing a (spacetime) point that may contain a particle of certain speed. Note that time is bi-infinite, since our cellular automata are reversible.

\item There is a \emph{Boolean particle} at spacetime position $(i, t, s)$ if $\pi_s(\hat f^t(x)_i)$ is not the constant 0 function, and a \emph{particle} if it is the constant 1 function.
%\item A \emph{pattern} is a map $P : D \to \F$, where $D \subset \Z \times \D$ is a finite domain of positions. We denote $D = D(P)$ and $|P| = |D(P)|$. The pattern $P$ \emph{occurs} in $x \in \hat A^\Z$ at time $t \in \Z$, if $\pi_s(\hat f^t(x)_i) = P(i, s)$ for all $(i, s) \in D(P)$.
\item There is a \emph{collision} at position $(i, t)$ if $\hat f^t(x)_i$ contains at least three Boolean particles, and a \emph{crossing} if there are at least two Boolean particles.
\item The \emph{input} is the pattern $x_{[0, n-1]} \in \hat A^n$.
\item The \emph{gadget} is the contents of $x$ outside $[0, n-1]$, an element of $A^{\Z \setminus [0, n-1]}$.
\item A \emph{line} is a subset of $\Z \times \N$ of the form $L = L(i, t, s) = \{ (i + s t, t) \;|\; t \in \N \}$ for some speed $s \in \D$ and $i \in \Z$. It is \emph{occupied (in a time interval $I \subset \Z$)} if one of its coordinates (in the region $\Z \times I$) contains a crossing or a Boolean particle of speed $s$. We denote $L^{(t)} = i + s t$ and $L^t = (L^{(t)}, t, s)$. The set of occupied lines in the spacetime diagram of interest is denoted $\Locc$.
\end{itemize}
\end{definition}

For example, there are three crossings in Figure~\ref{fig:SpacetimeLogic}, two of which are collisions. The highest intersection of two dashed lines is not a crossing, as it does not take place at an actual coordinate (one of the white circles). Every line segment in the figure defines an infinite occupied line.

%Dangerous line (for a direction) is a line containing at least one step of a particle going in that direction. Given a set of particles, or Boolean particles, we obtain some set of dangerous lines.

%Note that we consider a Boolean particle labeled by the constant zero function distinct from a nonexistent particle, to avoid having to take into account the ``accidental'' constant zero functions that may arise in our constructions. [WHY?]

\section{The Diffusion Lemma}

As stated above, we initialize the gadget to the all-0 partial configuration, in which situation we have the following lemma. It states that any finite set of particles in the CA $f$ eventually stop interacting and scatter to infinity. The corresponding result is proved for the physically universal CA of \cite{Sc14} by considering an \emph{abstract CA} over the alphabet $\{0, \frac{1}{2}, 1\}$, with the interpretation that $\frac{1}{2}$ can be either $0$ or $1$, and this lack of information is suitably propagated in collisions. In our version, the role the new state $\frac{1}{2}$ is played by Boolean particles.

\begin{lemma}[Diffusion Lemma]
\label{lem:Diffusion}
Let $x \in \hat A^\Z$ be such that $x_i = 0$ for all $i \notin [0, n-1]$. Then there are $O(n^2)$ crossings in the two-directional spacetime diagram of $x$ under $\hat f$ that all happen in a time window of length $O(n)$. For all other times $t \in \Z$, there are $O(n)$ Boolean particles in $\hat f^t(x)$, and they are contained in the interval $[-2|t|, n + 2|t|]$.
\end{lemma}

\begin{proof}
We prove the claim in the positive direction of time. By induction, one sees that after any $t \geq 0$ steps, there can be no right-going Boolean particles in the cells $(-\infty,t-1]$, and no left-going particles in the cells $[n-t, \infty)$. After these sets intersect at time $\lceil n/2 \rceil$, there are no collisions, so no new Boolean particles are created. Thus the number of Boolean particles is at most $6n$, that is, twice the length of the segment of $\hat f^{\lceil n/2 \rceil}(x)$ that may contain Boolean particles. We may also have $O(n^2)$ crossings between Boolean particles going in the same direction with different speeds. Thus there are $O(n^2)$ crossings in total. \qed
\end{proof}

\begin{figure}
\begin{center}
\begin{tikzpicture}[scale=1.15]

\draw (0,1) rectangle (8,5);

\draw[fill=black!20] (3.5,3) -- (0,4.75) -- (0,5) -- (1,5) -- (2,4.5) -- (1.5,5) -- (2.5,5) -- (4,3.5) -- (5.5,5) -- (6.5,5) -- (6,4.5) -- (7,5) -- (8,5) -- (8,4.75) -- (4.5,3) -- (8,1.25) -- (8,1) -- (7,1) -- (6,1.5) -- (6.5,1) -- (5.5,1) -- (4,2.5) -- (2.5,1) -- (1.5,1) -- (2,1.5) -- (1,1) -- (0,1) -- (0,1.25) -- cycle;
\draw[fill=black!30] (3.5,3) -- (4,3.5) -- (4.5,3) -- (4,2.5) -- cycle;
\draw[very thick] (3.5,3) -- (4.5,3);
\draw[dashed] (0,3) -- (3.5,3);
\draw[dashed] (4.5,3) -- (8,3);

\node (c) at (4,3) {};
\foreach \x/\l in {0/1.5, 1.75/1, 5.25/1, 6.5/1.5}{
  \foreach \y in {4.75, 1.25}{
    \draw[very thick] (\x,\y) -- (\x+\l,\y);
    \node[minimum size=.5cm] (a) at (\x+\l/2,\y) {};
    \draw[thick,dotted,->] (c) -- (a);
  }
}

\node [left] () at (0,3) {$x$};

\end{tikzpicture}
\end{center}
\caption{An illustration of Lemma~\ref{lem:Diffusion}. The dashed line is the configuration $x$, and the thick segment is the interval $[0, n-1]$. All collisions take place in the dark gray area, and Boolean particles may occur in the light gray area. The eight horizontal line segments are the scattering Boolean particles, grouped by speed.}
\label{fig:Diffusion}
\end{figure}
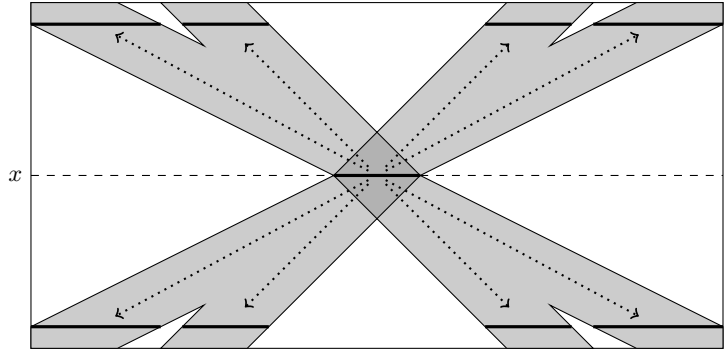

At this point, in the configuration of interest, we have an empty gadget, and the spacetime diagram contains $O(n)$ Boolean particles at any given time. Since the CA $\hat f$ is reversible, the values of the corresponding Boolean expressions determine the values of the original variables.

\section{Manipulating the Spacetime Diagram}

\subsection{Controlled Modifications}

In this section, we introduce new particles in the gadget that will collide with the existing Boolean particles and create new ones. This is called \emph{scheduling collisions}. We never schedule a collision on an occupied line, and never add a Boolean particle on an existing crossing. This is formalized in the following. %However, their main point is that any new collisions we have between particles must not be on top risky lines. Also, we never add any particles that go through dangerous cells. Because that's dangerous.

% a = O(1)
% b = O(m)

\begin{definition}
\label{def:Control}
A modification of the gadget is \emph{$(a, b, \Ls, \Ts, t, I)$-controlled} for numbers $a, b \in \N$, sets of lines $\Ls$ and $\Ts$, time $t \in \Z$, and interval $I \subset \Z$, if the following conditions hold:
\begin{enumerate}
\item the modification consists of adding at most $a$ particles to the gadget,
\item at most $a$ new occupied lines and $b$ new crossings are introduced,
\item no existing crossings become collisions,
\item no line in $\Ls \cup \Locc$ is occupied by a new Boolean particle or a new crossing,
\item no spacetime position in the set
\begin{equation}
\label{eq:Forbs}
F(\Ts, t, I) = \{ (i, t, s') \;|\; i \in I, s' \in D \} \setminus \Ts^t
\end{equation}
gets a new Boolean particle, and
\item no line in $\Ts$ gets a collision after time $t$.
\end{enumerate}
If the conditions hold in the time interval $(-\infty, t]$ (in particular, condition 6 need not hold at all), then the modification is \emph{weakly $(\Ls, \Ts, t, I)$-controlled}.
\end{definition}

In practice, a controlled modification is one where we add to the gadget a finite number of particles that affect the spacetime diagram of interest only where we want it to be affected: the spacetime positions in $\Ts^t$. The lines in $\Ls \cup \Locc$ and the positions near $\Ts$, that is, those in $F(\Ts, t, I)$, are `protected' from accidentally obtaining any auxiliary Boolean particles created in the modification.

\begin{definition}
Let $j, t \in \Z$. The \emph{positive cone rooted at $(j,t)$} is the set of spacetime coordinates $\mathcal{C}(j, t) = \{ (j+j', t+t') \;|\; t' \in \N, j' \in [-t', t'] \}$.
\end{definition}

The following lemmas are parametrized by the numbers $m_i$ for $i = 1, 2, 3, 4$. They denote, intuitively, the number of existing crossings, % this used to say collisions, but it's the number of crossings later
the number of occupied lines, the size of the additional `forbidden area,' and the number of particles involved, respectively. Also, the expression `if $t' = \Omega(N(m_1, m_2, m_3, m_4))$ is large enough, then $P(t')$ holds' for $t' \in \N$, a function $N : \N^4 \to \N$ and a property $P$ means that there exists a number $T \leq K \cdot f(m_1, m_2, m_3, m_4) + K'$, where $K, K' > 0$ are constants independent of the $m_i$, such that $P(t')$ holds for all $t' \geq T$.

\subsection{Moving the Boolean Particles}

We now prove that we can add a finite number of particles to the gadget, so that after some number of steps, a collection of Boolean particles is `moved' onto any desired lines, at the same time. We do this one particle at a time, and without interfering with the trajectories of any other existing particles.

% SP spacetime diagrams with the boolean particles as input
% LN the set of lines
% C(x) = positions of collisions in $x$
% indexing by a set of lines just means indexing by their union
%\begin{lemma}
%Let $x \in SP$, $L \subset LN$ with $|C(x)| \leq m$ $|L| = m$, $p \in LN \setminus L$ and $x_{LN \setminus (L \cup %\{p\})} = 0^{LN \setminus (L \cup \{p\})}$, so that $(j,t) \in p$ and $p$. Then there exists $y \in SP$ with $y|_L = x|_L$.
%\end{lemma}

% \begin{lemma}
% Suppose we have a spacetime diagram of interest with at most $m^2 \in \N$ crossings, and a Boolean particle $p$ with label $\beta \in \mathcal{F}$ at position $(j,s)$ at time $t$. Suppose further that the line $L$ that $p$ occupies does not contain any collision after time $t$. Let $\mathcal{L}$ be a collection of at most $m$ lines containing (at least) all the occupied lines, except that of $p$. Let $L' \notin \mathcal{L}$ be an unoccupied line of some slope $s'$ passing through position $(j',s')$ at time $t'$. Then if $t' > t + \Omega(m^2)$ and $j - t'-t \leq L^{t'} \leq j + t'-t$, we can add $O(1)$ particles to the gadget so that:
% \begin{itemize}
% \item no line in $\mathcal{L}$ gets occupied by a new Boolean particle or gets a new collision,
% \item no crossings become collisions,
% \item at most $O(m)$ new crossings are introduced,
% \item the number of new occupied lines is $O(1)$, and
% \item the position $(L')^{t'}, s')$ at time $t'$ contains a Boolean particle with label $\beta$.
% \end{itemize}
% \end{lemma}

\begin{lemma}
\label{lem:Movement}
Suppose we have a spacetime diagram of interest with $m_1$ crossings and $m_2$ occupied lines, and a Boolean particle $p$ with label $\beta \in \mathcal{F}$ on a line $L_p$ that contains no collisions after time $t$. Let $\Ls$ be a collection of $m_3$ lines not containing $L_p$. Let $L \notin \Ls$ be an unoccupied line that passes through some spacetime coordinate $(j', t') \in \mathcal{C}(j, t)$ with $t' > t$. Let $\Ts$ be a set of $O(m_3)$ lines containing $L$, and let $I \subset \Z$ be an interval of length $O(m_3)$. If $t' = t + \Omega(m_1 + m_2 + m_3)$ is large enough, then there is an $(O(1), O(m_2), \Ls, \Ts, t', I)$-controlled modification after which the spacetime position $L^{t'}$ contains a Boolean particle with label $\beta$. The same holds if the line $L$ is unoccupied only in the time interval $(-\infty, t']$, but the modification is weakly controlled.
\end{lemma}

\begin{proof}
%Let's move $p$ from position-time-speed $(j, t, 2)$ to position-time-speed $(j',t',-2)$ respecting a polynomial number of constraints.
Denote the speed of the particle $p$ by $s$, and the target speed by $s'$. We assume for simplicity that $s = 2$ and $s' = -2$.
%One can perform a similar analysis for moving from $(j, t, 2)$ to $(j',t',-1)$ and for moving from $(j, t, 1)$ to $(j',t',-2)$.
One can perform a similar analysis for the cases $(s, s') = (2, -1)$ and $(s, s') = (1, -2)$.
By mirror symmetry, and by repeating such a modification at most twice, one obtains all possible movements.

\begin{figure}
\begin{center}
\begin{tikzpicture}[scale=0.15]

\clip (-22.7,-5) rectangle (57,40);

% (j, t) = (\sa, \sb) and (j', t') = (\ta, \tb), k = k ofc...
\def\k{5}
\def\sa{6}
\def\sb{3}
\def\ta{20}
\def\tb{30}
\def\h{\sa + 0.333*\ta + 0.333*2*\tb + 0.333*\sb - 0.333*\sa - 0.333*\k - \sb + \k}
\def\m{0.333*\ta + 0.333*2*\tb + 0.333*\sb - 0.333*\sa - 0.333*\k}
\coordinate (O) at (\sa-2*\sb,0);
\coordinate (S) at (\sa,\sb);
\coordinate (H1) at ($ (S) + (2*\k,\k) $);
\coordinate (H2) at (\h, \m);
\coordinate (T) at (\ta,\tb);

% tikz/latex is a fucking piece of shit, the thing below is h-2m which of course doesn't work as such
\coordinate (P1) at (\h - 0.666*\ta - 0.666*2*\tb - 0.666*\sb + 0.666*\sa + 0.666*\k, 0);
% the thing below is h+m
\coordinate (P4) at (\h + 0.333*\ta + 0.333*2*\tb + 0.333*\sb - 0.333*\sa - 0.333*\k, 0);
\coordinate (P2) at (\sa+3*\k+\sb,0);
\coordinate (P3) at (\sa+4*\k+2*\sb,0);

\coordinate (T1) at (0,60);
\coordinate (T2) at (1,60);
%\path[name path=top] (0,25) -- (1,25);

\clip(-25,-5) rectangle (65,50);

%\node[fill,inner sep=2,shape=circle] at (0,0) {};

\node[below] at (O) {$(j-2t, 0)$};
\node[fill,inner sep=1,shape=circle] at (O) {};

\node[above left] at (S) {$(j, t)$};
\node[fill,inner sep=2,shape=circle] at (S) {};

\node[right=0.2] at (H1) {$(j + 2k, t + k)$};
\node[fill,inner sep=1,shape=circle] at (H1) {};

\node[right = 0.2] at (H2) {$(j + 2k + k', t + k + k')$}; % was {$(h, m)$};
\node[fill,inner sep=1,shape=circle] at (H2) {};

\node[below left] at (T) {$(j', t')$};
\node[fill,inner sep=2,shape=circle] at (T) {};

\node[below] at (P1) {$(j - 2t - k', 0)$}; % was {$(h-2m, 0)$};
\node[fill,inner sep=1,shape=circle] at (P1) {};
\node[below] at (P2) {$(j+t+3k, 0)$};
\node[fill,inner sep=1,shape=circle] at (P3) {};
\node[above right] at (P3) {$(j+2t+4k, 0)$};
\node[fill,inner sep=1,shape=circle] at (P4) {};
\node[below left] at (P4) {$(j + t + 3k + 2k', 0)$}; % was {$(h+m, 0)$};
\node[fill,inner sep=1,shape=circle] at (P2) {};

\draw (P1) -- (H2);
\draw (P4) -- (H2);
\draw (P2) -- (H1);
\draw (P3) -- (H1);
\draw[dashed] (O) -- (H1) -- (H2) -- (T);

% then the unavoidable useless Boolean particles...

\coordinate (T3) at (intersection of T1--T2 and P3--H1);
\draw (H1) -- (T3);
\coordinate (T4) at (intersection of T1--T2 and P2--H1);
\draw[dashed] (H1) -- (T4);
\coordinate (T4) at (intersection of T1--T2 and P2--H1);
\draw[dashed] (H1) -- (T4);
\coordinate (T5) at (intersection of T1--T2 and H2--T);
\draw[dashed] (T) -- (T5);
\coordinate (T6) at (intersection of T1--T2 and H1--H2);
\draw[dashed] (H2) -- (T6);
\coordinate (T7) at (intersection of T1--T2 and P1--H2);
\draw[dashed] (H2) -- (T7);
\coordinate (T8) at (intersection of T1--T2 and O--H1);
\draw[dashed] (H1) -- (T8);
\coordinate (T9) at (intersection of T1--T2 and P4--H2);
\draw (H2) -- (T9);

% then the unavoidable collisions of those particles, except the intersection D of L2 and L6 which is invisible

\coordinate (A) at (intersection of P1--H2 and P3--H1);
\node[left=0.2] at (A) {$A$};
\node[fill, inner sep=1,shape=circle] at (A) {};

\coordinate (B) at (intersection of P2--H1 and P1--H2);
\node[right=0.2] at (B) {$B$};
\node[fill, inner sep=1,shape=circle] at (B) {};

\coordinate (C) at (intersection of P4--H2 and O--H1);
\node[right=0.2] at (C) {$C$};
\node[fill, inner sep=1,shape=circle] at (C) {};

% placing names of lines...

\node[below right=-0.1] () at ($ (O)!0.5!(H1) $) {$L_0$};
\node[above left=-0.1] () at ($ (P1)!0.5!(A) $) {$L_1$};
\node[below left=-0.1] () at ($ (P2)!0.5!(H1) $) {$L_2$};
\node[above right=-0.1] () at ($ (P3)!0.5!(H1) $) {$L_3$};
\node[above right=-0.1] () at ($ (P4)!0.5!(H2) $) {$L_4$};
\node[below right=-0.1] () at ($ (H1)!0.6!(H2) $) {$L_5$};
\node[below left=-0.1] () at ($ (H2)!0.4!(T5) $) {$L_6$};

%\path [name path=derp] (P2) -- (H1);
%\path [name intersections={of=derp and top}];
%\node[fill,inner sep=1,shape=circle] at (intersection of T1--T2 and P2--H1) {};

%\node[fill,inner sep=1,shape=circle] at (2*\h-2*\m-\sa+2*\sb, \h-\m-\sa-2*\sb) {};

\end{tikzpicture}
\end{center}
\caption{Moving a Boolean particle.}
\label{fig:Moving}
\end{figure}
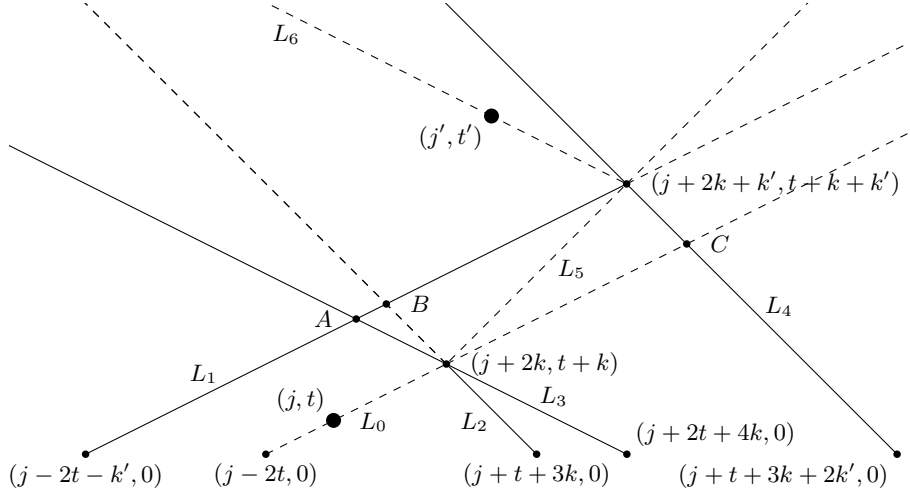

We change the Boolean particle $p$ to a slow right Boolean particle after some $k \in \N$ steps, by scheduling a suitable collision at the spacetime coordinate $(j+2k, t+k)$. After some $k' \in \N$ more steps, we turn this Boolean particle into a fast left one; this happens at the spacetime coordinate $(j+2k+k', t+k+k')$. Our goal is that $t' - t - k - k'$ steps after the second collision, the Boolean particle is at $(j', t')$. Of course, whether this happens depends on our choices of $k$ and $k'$. Namely, we must have
\[ j' = j + 2k + k' - 2(t' - t - k - k') = j + 4k + 3k' + 2t - 2t' \]
due to the speeds of the Boolean particles. This implies that
\begin{equation}
\label{eq:kprime}
% j' = j+4k+3k'+2t-2t'
% j'-j-4k-2t+2t'=3k'
k' = \frac13(j' - j + 2(t' - t) - 4k)
\end{equation}
so that in particular $j' - j + 2(t' - t) - 4k$ should be divisible by $3$, or in other words, $k \equiv 4k \equiv j' - j + 2(t' - t) \bmod 3$.

% (h,m) = (j+2k+k',t+k+k')

%We want to turn this into a fast left particle going through $(j',t')$. We do this in some coordinate $(h, m) \in \Z \times \N$. Whether this is possible depends on our choice of $k$, of course. Namely, we must have
% \[ j' = h - 2(t' - m) \]
% and
% \[ h = j+2k+(m-t-k) = j+m-t+k, \]
% due to the speeds of the particles. Among pairs of reals $(h,m)$, the unique solution is
% m = h-j+t-k = j'+2t'-2m-j+t-k \implies m = (1/3)(j'+2t'-j+t-k)
% \[ (h, m) = (j+m - t + k, \frac{1}{3}(j' + 2t'+t-j-k)). \]
% This is a pair of integers if $k \equiv a = j' + 2t' + t - j \bmod 3$.

The other obvious constraints are $0 < k$, $0 < k'$ and $k + k' \leq t' - t$. Expanding~\eqref{eq:kprime} in the latter two inequalities, we obtain
% 0 < j'-j+2(t'-t)-4k
% 4k < j'-j+2(t'-t)
% k < 1/4(j'-j+2(t'-t))
%%%
% k + k' = k + 1/3(j'-j+2(t'-t)-4k) < t'-t
% 3k+j'-j+2(t'-t)-4k < 3(t'-t)
% -k < j-j'+t'-t
% k > j'-j+t-t'
\begin{equation}
\label{eq:kbounds}
\max(0, j' - j - (t' - t)) \leq k < \frac14(j' - j + 2(t' - t))
\end{equation}
for the variable $k$. Denote by $I \subset \Z$ this interval of possible values for $k$.
% 1/4(j'-j+2(t'-t)) - (j'-j-(t'-t)) = 1/4(j'-j+2t'-2t-4j'+4j+4t'-4t) = 3/4(j-j'-2(t-t'))

% We also want $t+k < m < t'$, that is,
% \[ t+k < \frac{1}{3}(j' + 2t'+t-j-k) < t'. \]
% %The first inequality gives
% %3t+3k < j'+2t'+t-j-k \iff
% %\[ k < \frac{1}{4}(j'+2t'-2t-j), \]
% %and the second one gives
% % j' + 2t'+t-j-k < 3t'
% % j' +2t' + t - j - 3t' < k
% %\[ j' -t' + t - j - 3t' < k. \]
% Solving both inequalities for $k$ gives the constraint
% \[ k \in I = (j' + t - j - 4t', \frac{1}{4}(j'+2t'-2t-j)). \]
% The length of this interval is
% % \[ \frac{1}{4} j'+\frac{1}{2}t'-\frac{1}{2}t-\frac{1}{4}j - j' - t + j + 4t' -1 = \]
% \[ a = \frac{3}{2}(3t' - t) + \frac{3}{4}(j - j') - 1. \]

By the assumptions, we have $j - (t'-t) \leq j' \leq j + (t'-t)$, which implies $j' - j - (t' - t) \leq 0$, so that the left hand side of~\eqref{eq:kbounds} can be replaced by $0$. The length of the interval $I$ is therefore
\[ a = \frac14 (j - j' + 2(t' - t)) \geq \frac14 (t'-t) = \Omega(m_1+m_2+m_3). \]
Since the only other requirement on $k$ is parity modulo $3$, there are $\Omega(m_1+m_2+m_3)$ possible choices for $k$. If we have $k \in I$ and $k \equiv j' - j + 2(t' - t) \bmod 3$, we say \emph{$k$ respects the spacetime constraints}.

% By the assumptions, we have $t' > t + C m^2$ and $j - (t'-t) \leq j' \leq j + (t'-t)$, where $C$ is some large constant. Then
% \begin{align*}
% a &= \frac{3}{2}(3t' - t) + \frac{3}{4}(j - j') - 1 \\
% &> 3t' + \frac{3}{2}(t' - t) - \frac{3}{4}|j - j'| - 1 \\
% %&\geq 3t' + \frac{6}{4}(t' - t) - \frac{3}{4}(t' - t) - 1 \\
% &\geq 3t' + \frac{3}{4}(t' - t) - 1 \\
% &\geq 3 C m^2 + \frac{3}{4}(Cm^2) - 1 \\
% &\geq C' m^2 \\
% \end{align*}
% where $C' = 3C + \frac{3}{4} C - 1$, another large constant. Since the only other requirement on $k$ is parity modulo $3$, erring on the safe side, there are at least $C'' m^2 = (C' / 4) m^2$ possible choices for $k$. If $k \in I$ and $k = a \bmod 3$, we say \emph{$k$ respects the spacetime constraints}.

Now, for any $k$ respecting the spacetime constraints, we give a set of particles that, when added to the gadget at time $0$, realize these two collisions assuming that they do not collide with any existing particles. The positions (coordinate-speed tuples) of the particles are $P_1 = (j - 2t - k', 2)$, $P_2 = (j + 3k + t, -1)$, $P_3 = (j + 4k + 2t, -2)$ and $P_4 = (j + t + 3k + 2k', -1)$. We denote by $L_i$ the line containing the position $P_i$ at time $0$, and also $L_p = L_0$.

%Now, for any $k$ respecting the spacetime constraints, let us compute a set of cells and particle types such that, if added to $x$, they realize these two collisions assuming that they do not collide with any existing particles. As shown in the figure, we should add the following particles:
%\begin{itemize}
%\item a fast right-going particle at $P_1 = (h-2m, 0)$,
%\item a slow left-going particle at $P_2 = (j+3k+t, 0)$,
%\item a slow right-going particle at $P_3 = (j + 4k + 2t, 0)$, and
%\item a fast left-going particle at $P_4 = (h+m, 0)$.
%\end{itemize}
%Call these particles $L_1, L_2, L_3, L_4$ respectively (or their lines I guess).

First, let us forget about the lines $\Ls$ and the existing Boolean particles other than $p$, and track how $p$ (continued backward to its imaginary starting position on $x$, although it may of course have been produced in a collision), and the new particles added at $P_1$, $P_2$, $P_3$, $P_4$, interact. We show that, in fact, they move the particle $p$ to $(j', t')$, as desired for any choice of $k$ that respects the spacetime constraints, and all in all, no new collisions other than the two desired ones are introduced. Of course, once we have done this, it is enough to show that for some choice of $k$ that respects the spacetime constraints, the conditions of Definition~\ref{def:Control} are satisfied, since then the new Boolean particles introduced by our modification do not interfere with the existing ones in unwanted ways.

% none of the new collisions take place on a line of $\mathcal{L}$ or one of the existing $m$ crossings. If we can choose $k$ like this, then we have proved the lemma (note that the fact that only $O(m)$ new crossings are introduced is automatic, since a Boolean particle can only intersect a given line in $\mathcal{L}$ once).

Let us proceed to the analysis. The reader should consult Figure~\ref{fig:Moving} for a visualization. At the spacetime coordinate $(j+2k, t+k)$, the Boolean particle $p$ collides with the particles added to $P_2$ and $P_3$, resulting in all four types of Boolean particles. The new occupied line we are interested in is the path where we are moving $p$, or the slope-$1$ line $L_5$. At the spacetime coordinate $(j + 2k + k', t + k + k')$, where $P_3$, $P_4$ and the slowed-down $p$ collide, all four Boolean particles are produced as well. Again, the interesting line is the one where we move $p$, or the line $L = L_6$ with slope $-\frac12$.

% (h,m) = (j+2k+k',t+k+k')
%By a glance at the picture, it is clear that things work out as planned, at least for some choices of $k$.
Let us compute the equations of the lines $L_0, L_1, \ldots, L_6$:
\begin{align*}
 L_0 &= \{(j-2t+2y, y) \;|\; y \in \N\}, & L_1 &= \{(j-2t-k'+2y, y) \;|\; y \in \N\}, \\
 L_2 &= \{(j+t+3k-y, y) \;|\; y \in \N\}, & L_3 &= \{(j+2t+4k-2y, y) \;|\; y \in \N\}, \\
 L_4 &= \{(j+t+3k+2k'-y, y) \;|\; y \in \N\}, & L_5 &= \{(j-t+k+y, y) \;|\; y \in \N\}, \\
 L_6 &= \{(j'+2t'-2y, y) \;|\; y \in \N\}. & &
\end{align*}
%\[ L = \{(j-2t+2y, y) \;|\; y \in \N\},\]
%\[ L_1 = \{(h-2m+2y, y) \;|\; y \in \N\},\]
%\[ L_2 = \{(j+3k+t-y, y) \;|\; y \in \N\},\]
%\[ L_3 = \{(j+3k+t-y, y) \;|\; y \in \N\},\]
%\[ L_4 = \{(h+m-y, y) \;|\; y \in \N\},\]
%\[ L_5 = \{(j+k-t+y, y) \;|\; y \in \N\},\]
%and
%\[ L' = L_6 = \{(h+2m-2y, y) \;|\; y \in \N\}.\]
Now, let us compute their intersection points (in addition to the desired intersections $(j+2k,t+k)$ and $(j+2k+k',t+k+k')$) into Table~\ref{tab:Intersect}. The last two intersections, $D$ and $E$, are not shown in Figure~\ref{fig:Moving}, since $D$ did not fit, and $E$ does not appear in the configuration (the others always appear).

\begin{table}[ht]
\begin{center}
\caption{The incidental intersections of the lines $L_0, L_1, \ldots, L_6$.}
\label{tab:Intersect}
\renewcommand{\arraystretch}{1.2}
\begin{tabular}{|c|c|c|}
\hline
Lines & Intersection point & Label \\
\hline \hline
$L_1$ and $L_3$ & $(j+2k-\frac{1}{2}k', t+k+\frac{1}{4}k')$ & $A$ \\
\hline
$L_1$ and $L_2$ & $(j+2k-\frac{1}{3}k',t+k+\frac{1}{3}k')$ & $B$ \\
\hline
$L_0$ and $L_4$ & $(j+2k+\frac{4}{3}k', t+k+\frac{2}{3}k')$ & $C$ \\
\hline
$L_2$ and $L_6$ & $(j+2k+\frac{3}{2}k', t+k+\frac{3}{4}k')$ & $D$ \\
\hline
$L_3$ and $L_4$ & $(j+2k+4k', t+k-2k')$ & $E$ \\
\hline
\end{tabular}
\end{center}
\end{table}

% The lines $L_1$ and $L_3$ intersect at
% \[ A = (j+2k-\frac{k'}{2}, t+k+\frac{k'}{4}). \]
% The lines $L_2$ and $L_1$ intersect at
% \[ B = (j+2k-\frac{k'}{3},t+k+\frac{k'}{3}). \]
% The lines $L_0$ and $L_4$ intersect at
% \[ C = (j+2k+\frac{4k'}{3}, t+k+\frac{2k'}{3}). \]
% The lines $L_2$ and $L_6$ intersect at
% \[ D = (j+2k+\frac{3k'}{2}, t+k+\frac{3k'}{4}). \]
% The lines $L_3$ and $L_4$ may intersect a
% \[ E = (j+2k+4k', t+k-2k'). \]
% No other intersections exist between these lines.

%The lines $L_1$ and $L_3$ intersect at
%\[ A = (\frac{1}{3}h + \frac{2}{3}j - \frac{2}{3}m + \frac{2}{3}t, -\frac{1}{3}h + \frac{2}{3}m + \frac{1}{3}j + k + \frac{1}{3}t). \]
%The lines $L_2$ and $L_1$ intersect at
%\[ B = (j+2k,k+t). \]
%The lines $L$ and $L_4$ intersect at
%\[ C = (\frac{2}{3}h + \frac{1}{3}j + \frac{2}{3}m - \frac{2}{3}t, \frac{1}{3}h - \frac{1}{3}j + \frac{1}{3}m + \frac{2}{3}t). \]
%The lines $L_2$ and $L_6$ intersect at
%\[ D = (-h+2j+6k+2m+2t,h-j-3k-2m-t). \]
%The lines $L_3$ and $L_4$ may intersect a
%\[ E = (2h-j-4k+2m-2t,-h+j+4k-m+2t). \]
%No other intersections exist between these lines.

It is easy to verify from the equations of the lines that for each $i \in \{1, 2, 3, 4, 5\}$, every choice of $k$ results in a different line $L_i$ (the lines $L_0$ and $L_6$ are always the same). Also, the two collisions we introduced are different for each $k$. Recall now that there are $m_1$ crossings and $m_2$ occupied lines in the spacetime diagram of interest, and that the set $\Ls$ has cardinality $m_3$. Furthermore, the cardinality of the set $\Ts$ of target lines and of the set $F(\Ts, t', I)$ of forbidden spacetime positions defined in \eqref{eq:Forbs} is likewise $O(m_3)$. Thus, there are $O(m_1+m_2+m_3)$ such choices for $k$ that at least one of the lines $L_i$ is in $\Ls \cup \Locc$ or contains an existing crossing, a line in $\Ls \cup \Locc$ gets a new collision, or a spacetime position in $F(\Ts, t, I)$ gets a Boolean particle. If $a = \Omega(m_1+m_2+m_3)$ is large enough, we can guarantee the existence of a $k$ that respects the spacetime constraints, and is not in the aforementioned set of cardinality $O(m_1+m_2+m_3)$.

For such $k$, the modification of adding the particles $P_1$, $P_2$, $P_3$ and $P_4$ to the gadget satisfies the conditions 3, 4 and 5 of Definition~\ref{def:Control} in the time interval $(-\infty, t']$. Condition 2 is also satisfied in this interval, since the only new occupied lines are $L_1, \ldots, L_6$, and every new crossing happens at the intersection point of an existing occupied line and one of $L_1, \ldots, L_6$, the number of which is $O(m_2)$, or at one of the five intersection points in Table~\ref{tab:Intersect}. Finally, condition 1 is satisfied by definition, and thus the modification is weakly $(O(1), O(m_2), \Ls, \Ts, t', I)$-controlled, and results in the spacetime position $(j', t', -2)$ containing a Boolean particle with label $\beta$.

If the target line $L$ is completely unoccupied, then condition 3 is satisfied in the entire spacetime diagram, so that no new collisions or occupied lines are introduced after time $t'$. We can also check from Table~\ref{tab:Intersect} that the only crossing on line $L$ is at $(j+2k+k', t+k+k')$, so that condition 6 is satisfied. Then the modification is $(O(1), O(m_2), \Ls, \Ts, t', I)$-controlled. \qed
\end{proof}

The point of the set of forbidden lines $\Ls$ and positions $I$ is that  we can in fact move an arbitrary number of Boolean particles at the same time. The idea is that we move them by applying Lemma~\ref{lem:Movement} repeatedly to one Boolean particle at a time, always adding the target lines of the remaining ones into the protected set $\Ls$. This guarantees that the lines are not accidentally occupied too early. %The target set $T$ is now fixed as the set of desired positions of the particles.

\begin{corollary}
\label{cor:Movement}
Suppose we have a spacetime diagram of interest with $m_1$ crossings and $m_2$ occupied lines, and some Boolean particles $p_1, \ldots, p_{m_4}$ with labels $\beta_k \in \mathcal{F}$ on lines $L_{p_k}$ that contain no collisions after time $t$. Let $\Ls$ be a collection of $m_3$ lines not containing any of the $L_{p_k}$. Let $L_k \notin \Ls$ be unoccupied and mutually disjoint lines that pass through some spacetime coordinates $(j'_k, t') \in \mathcal{C}(L_{p_k}^{(t)}, t)$ with $t' > t$. Denote $\Ts = \{ L_1, \ldots, L_{m_4} \}$, and let $I \subset \Z$ be an interval of length $O(m_3)$. If $t' = t + \Omega(m_1 + m_3 + m_4(m_2 + m_4))$ is large enough, then there is an $(O(m_4), O(m_4(m_2 + m_4)), \Ls, \Ts, t', I)$-controlled modification after which each spacetime position $L_k^{t'}$ contains a Boolean particle with label $\beta_k$. The same holds if the lines $L_k$ are unoccupied only in the time interval $(-\infty, t']$, but the modification is weakly controlled.
\end{corollary}

\begin{proof}
We apply Lemma~\ref{lem:Movement} to each Boolean particle $p_k$ separately. First, let $\Ls_1 = \Ls \cup \{ L_2, \ldots, L_{m_4} \}$. We can do a (weakly) $(O(1), O(m_2), \Ls_1, \Ts, t, I)$-controlled modification to the gadget such that the spacetime position $L_1^{t'}$ contains a Boolean particle with label $\beta_1$. This modification does not introduce crossings or Boolean particles to the lines of $\Ls_1$ (in particular, it does not affect the Boolean particles $p_2, \ldots, p_{m_4}$), and it creates $O(m_2)$ new crossings and a constant number of new occupied lines.

We continue inductively for $k < m_4$, assuming that the particles $p_1, \ldots, p_k$ have been moved to their places by introducing $O(k m_2 + k^2)$ new crossings and $O(k)$ new occupied lines. We now define $\Ls_{k+1} = \Ls \cup \{ L_{k+2}, \ldots, L_m \}$, and use Lemma~\ref{lem:Movement} to do a (weakly) $(O(1), O(m_2 + k), \Ls_{k+1}, \Ts, t, I)$-controlled modification to the gadget, which causes $L_{k+1}^{t'}$ to contain a Boolean particle with label $\beta_{k+1}$. Note that $\Ls_{k+1}$ has size $O(m_3 + m_4)$, there are $O(m_1 + k m_2 + k^2)$ crossings and $O(m_2 + k)$ occupied lines in total, and $t' = t + \Omega(m_1 + m_3 + m_4(m_2 + m_4))$. Furthermore, if the lines $L_k$ do not intersect at all, no new crossings are introduced on them at any point. Thus, the conditions of Lemma~\ref{lem:Movement} hold.

After the inductive construction, we have modified the gadget by adding $O(m_4)$ particles to it. The total numbers of new crossings and occupied lines are $O(m_4(m_2 + m_4))$ and $O(m_4)$, respectively. The remaining conditions for the modification to be (weakly) $(O(m_4), O(m_4(m_2 + m_4)), \Ls, \Ts, t', I)$-controlled are easy to verify. This concludes the proof. \qed
\end{proof}

%In the initial configuration of interest, we have the input of length $n$, containing $4n$ Boolean particles, each over a different variable. First, we wait for them to disperse: by Lemma~\ref{lem:Diffusion}, after $t_1 = O(n)$ steps we have no crossings, and there are $O(n)$ Boolean particles. The number of crossings in the spacetime diagram of interest is $O(n^2)$. By Corollary~\ref{cor:Movement}, we can add $O(n)$ particles to the gadget, and reposition the Boolean particles so that they are moving to the right with speed $1$ and occupy some interval of length $O(n)$ at time $t_2 = O(n^3)$.

\subsection{Computing with the Boolean Particles}

Next, we will do some computation with the Boolean particles. Namely, we show that the NAND of two Boolean particles can be computed nondestructively in any spacetime position, as long as we have enough time, and the target line is in the intersection of the cones rooted at the input particles.% The statements of the following lemmas are slightly simpler than above, since we have no need for weakly controlled modifications here.

\begin{lemma}
\label{lem:Computation}
Suppose we have a spacetime diagram of interest with $m_1$ crossings and $m_2$ occupied lines, and two speed-$1$ Boolean particles $p_1, p_2$ labeled $\beta_1, \beta_2 \in \F$ with $\beta_2 \wedge \beta_1 \not\equiv 1$ which occupy distinct lines $L_1$ and $L_2$ that contain no collisions after time $t$. Let $\Ls$ be a set of $m_3$ lines not containing $L_1$ and $L_2$, and let $L \notin \Ls$ be an unoccupied line of slope $1$ to the left of $L_1$ and $L_2$ that passes through some spacetime coordinate $(j', t') \in \mathcal{C}(L_1^{(t)}, t) \cap \mathcal{C}(L_2^{(t)}, t)$ with $t' > t$. Let $\Ts = \{L_1, L_2, L\}$. If $t' = t + \Omega(m_1 + m_2 + m_3)$ is large enough, then there is an $(O(1), O(m_2), \Ls, \Ts, t', \emptyset)$-controlled modification after which the spacetime positions $L_1^{t'}$, $L_2^{t'}$ and $L^{t'}$ contain Boolean particles labeled $\beta_1$, $\beta_2$ and $\neg (\beta_1 \wedge \beta_2)$, respectively.
\end{lemma}

The condition $\beta_2 \wedge \beta_1 \not\equiv 1$ is there mainly for terminological reasons, ensuring that the output of the computation is actually a Boolean particle, and does not restrict the usefulness of the lemma, since the constant $0$ function is available for us anyway. Similarly to the fact that any number of particles can be moved, we can now compute an arbitrary Boolean function, given enough time and space.

\begin{proof}
We proceed as in the proof of Lemma~\ref{lem:Movement}, but skipping most of the details. For $i \in \{1, 2\}$, we will schedule a collision at the spacetime coordinate $(j_1+k_1, t+k_1)$ for some $k_1 \leq k_2 \in \N$, in which a slow left-moving and a fast right-moving particle collide with $p_i$, creating a fast left-moving Boolean particle with label $\beta_i$ and leaving $p_i$ intact. At time $t + k_2 + k$ for some $k \in \N$, we schedule a third collision that slows down the copy of $p_1$. After some time, the still fast copy of $p_2$ will catch it, and we schedule a fourth collision at this point that creates a fast right-moving Boolean particle with label $\beta_1 \wedge \beta_2$. A fifth collision will be scheduled at the spacetime position where this particle reaches the target line $L$, producing a slow right-moving Boolean particle with label $\neg (\beta_1 \wedge \beta_2)$.

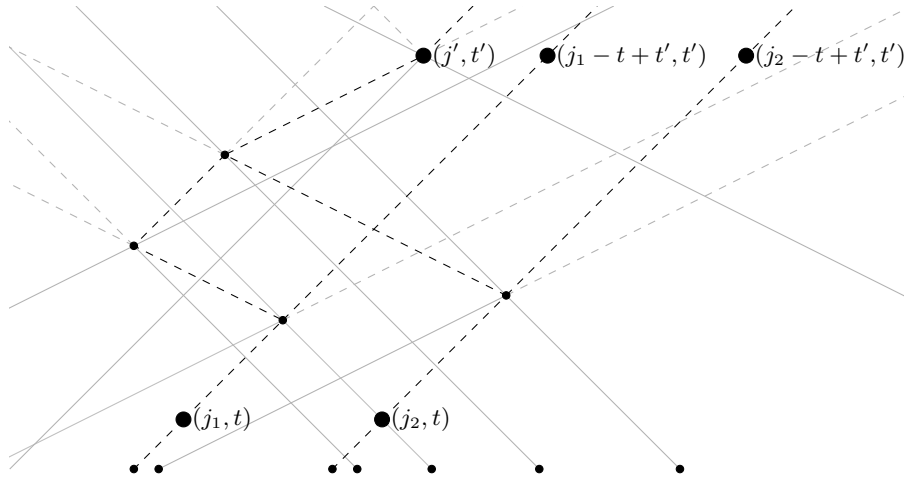
\begin{figure}[ht]
\begin{center}
\begin{tikzpicture}[scale=.33]

% Independent parameters
\def\ja{5} % j_1
\def\jb{13} % j_2
\def\jc{0} % point on target line at t
\def\ka{4} % k_1
\def\kb{5} % k_2
\def\k{3} % k
\def\t{2} % t
\def\tf{2} % time from t' to end

% Dependent parameters
\def\kka{-4*\ja/3+\jb/3+\jc-\ka+\kb+4*\k+\tf} % time from t+k_1 to end
\def\kkb{-4*\ja/3+\jb/3+\jc+4*\k+\tf} % time from t+k_2 to end
\def\tp{-\t+4*\ja/3-\jb/3-\jc-\kb-4*\k} % -t'
\def\tff{\t-4*\ja/3+\jb/3+\jc+\kb+4*\k+\tf} % end

\clip (\jc-2, -.5) rectangle (\jc+34, \tff);

% Coordinates of collisions and Boolean particles
\coordinate (p1) at (\ja,\t);
\coordinate (p2) at (\jb,\t);
\coordinate (c1) at (\ja+\ka,\t+\ka);
\coordinate (c2) at (\jb+\kb,\t+\kb);
\coordinate (c3) at (\ja+\ka-2*\k,\t+\ka+\k);
\coordinate (c4) at (2*\ja/3+\jb/3+\kb-2*\k, \t-\ja/3+\jb/3+\kb+\k);
\coordinate (c5) at (-4*\ja/3+\jb/3+2*\jc+\kb+4*\k, \t-4*\ja/3+\jb/3+\jc+\kb+4*\k);
\coordinate (q1) at ($ (c5) + (\ja-\jc,0) $);
\coordinate (q2) at ($ (c5) + (\jb-\jc,0) $);

% Coordinates of auxiliary particles
\coordinate (top1) at (0,\tff);
\coordinate (top2) at (1,\tff);

\coordinate (a11) at (\ja+2*\ka+\t,0);
\coordinate (a12) at (\ja-\ka-2*\t,0);
\coordinate (a21) at (\jb+2*\kb+\t,0);
\coordinate (a22) at (\jb-\kb-2*\t,0);
\coordinate (a31) at (\ja-\ka-4*\k-2*\t,0);
\coordinate (a32) at (\ja+2*\ka-\k+\t,0);
\coordinate (a4)  at (\ja/3+2*\jb/3+2*\kb-\k+\t,0);
\coordinate (a51) at (-8*\ja/3+2*\jb/3+3*\jc+2*\kb+8*\k+\t,0);
\coordinate (a52) at (-4*\ja+\jb+4*\jc+3*\kb+12*\k+2*\t,0);

% Paths of auxiliary particles
\draw[black!25!white] (a11) -- (c1) -- (intersection of a11--c1 and top1--top2);
\draw[black!25!white] (a12) -- (c1);
\draw[black!30!white,dashed] (c1) -- (intersection of a12--c1 and top1--top2);
\draw[black!30!white] (a21) -- (c2) -- (intersection of a21--c2 and top1--top2);
\draw[black!30!white] (a22) -- (c2);
\draw[black!30!white,dashed] (c2) -- (intersection of a22--c2 and top1--top2);
\draw[black!30!white] (a31) -- (c3) -- (intersection of a31--c3 and top1--top2);
\draw[black!30!white] (a32) -- (c3);
\draw[black!30!white,dashed] (c3) -- (intersection of a32--c3 and top1--top2);
\draw[black!30!white,dashed] (c3) -- (intersection of c1--c3 and top1--top2);
\draw[black!30!white] (a4)  -- (c4) -- (intersection of a4--c4 and top1--top2);
\draw[black!30!white,dashed] (c4) -- (intersection of c2--c4 and top1--top2);
\draw[black!30!white,dashed] (c4) -- (intersection of c3--c4 and top1--top2);
%\draw[black!30!white] (a51) -- (c5);
\draw[black!30!white,dashed] (c5) -- (intersection of a51--c5 and top1--top2);
\draw[black!30!white] (a52) -- (c5) -- (intersection of a52--c5 and top1--top2);
\draw[black!30!white] (c5) -- ++(\tp,\tp);
\draw[black!30!white,dashed] (c5) -- (intersection of c4--c5 and top1--top2);

% Bullets for collisions and Boolean particles
\draw[fill] (p1) circle (0.3);
\draw[fill] (p2) circle (0.3);
\draw[fill] (q1) circle (0.3);
\draw[fill] (q2) circle (0.3);
\draw[fill] (\ja-\t,0) circle (0.15);
\draw[fill] (\jb-\t,0) circle (0.15);
\draw[fill] (c1) circle (0.15);
\draw[fill] (c2) circle (0.15);
\draw[fill] (c3) circle (0.15);
\draw[fill] (c4) circle (0.15);
\draw[fill] (c5) circle (0.3);

% The actual computation
\draw[dashed] (p1) -- (c1) -- (c3) -- (c4);
\draw[dashed] (p2) -- (c2) -- (c4) -- (c5);
\draw[dashed] (c1) -- ++(\kka,\kka);
\draw[dashed] (c2) -- ++(\kkb,\kkb);
\draw[dashed] (p1) -- ++(-\t,-\t);
\draw[dashed] (p2) -- ++(-\t,-\t);
\draw[dashed] (c5) -- ++(\tf,\tf);

% Auxiliary particles
\foreach \l in {a11,a12,a21,a22,a31,a32,a4,a52}{
	\draw[fill] (\l) circle (0.15);
}

% Labels
\node [right] () at (p1) {$(j_1, t)$};
\node [right] () at (p2) {$(j_2, t)$};
\node [right] () at (q1) {$(j_1-t+t', t')$};
\node [right] () at (q2) {$(j_2-t+t', t')$};
\node [right] () at (c5) {$(j', t')$};

\end{tikzpicture}
\end{center}
\caption{Computing the NAND of two Boolean particles. The auxiliary Boolean particles are shown in gray for clarity.}
\label{fig:Computation}
\end{figure}

Figure~\ref{fig:Computation} shows the auxiliary particles needed for the collisions. It also shows that, for suitable values of $k_1$, $k_2$ and $k$, the only collisions resulting from the modification are the five desired ones. As in Lemma~\ref{lem:Movement}, there are then a bounded number of new occupied lines and $O(m_2)$ new crossings, and if $t'$ is large enough, we can find such values for the parameters that the conditions of Definition~\ref{def:Control} are satisfied.
\qed
\end{proof}

\begin{corollary}
\label{cor:Computation}
Suppose we have a spacetime diagram of interest with $m_1$ collisions and $m_2$ occupied lines, a coordinate $j \in \Z$, and Boolean particles $p_1, \ldots, p_{m_4}$ with labels $\beta_k \in \F$ at the spacetime positions $(j + k - 1, t, 1)$. Suppose further that the lines $L_{p_k}$ that the $p_k$ occupy do not contain collisions after time $t$. Let $H : \{0, 1\}^{m_4} \to \{0, 1\}^m$ with $m \leq m_4$ be a Boolean function realizable with $C$ NAND-gates such that $\pi_\ell \circ H \not\equiv 0$ for all $\ell = 1, \ldots, m$. Let $\Ls$ be a set of $m_3$ lines not containing the $L_{p_k}$, let $L_1, \ldots, L_{C + m} \notin \Ls$ be unoccupied lines of slope $1$ that pass through the spacetime coordinates $(j-1, t), \ldots, (j-C-m, t)$, and let $t' > t$. Let also $\Ts = \{L_C, \ldots, L_{C + m}\}$. If
\[ t' = t + \Omega((C + m_4)(C^2 + m_1 + m_3 + (C + m_4)(m_2 + m_4))) \]
is large enough, then there is an $(O(C + m_4), O((C + m_2 + m_4)(C + m_4)), \Ls, \Ts, t', \emptyset)$-controlled modification after which the spacetime positions $L_{C + \ell}^{t'}$ contain Boolean particles with labels $H(\beta_1, \ldots, \beta_m)_\ell$.
\end{corollary}

\begin{proof}
This proof is similar to that of Corollary~\ref{cor:Movement}. The $C$ auxiliary lines are used for the computation of $H$ one NAND-gate at a time. Let $t_0 = t$, $m_1^{(0)} = 0$ and $m_2^{(0)} = 0$.

For each $i = 1, \ldots, C + m$, we apply Lemma~\ref{lem:Computation} to some pair of particles $p_{i,1}$ and $p_{i,2}$ on some of the lines $L_{p_k}$ for $k \in \{1, \ldots, m_4\}$ or $L_k$ for $k < i$, and the target line $L_i$, denoting the set of these three lines by $\Ts_i$. We denote the number of new crossings after round $i-1$ of the calculation by $m_1^{(i-1)}$, and that of new occupied lines by $m_2^{(i-1)}$. We define $\Ls_i = \Ls \cup \{L_{i+1}, \ldots, L_{m+C}\}$, so that the latter lines will not accidentally become occupied, and then $|\Ls_i| = O(m_3) + C+m-i$. For large enough
\begin{equation}
\label{eq:tdiff}
t_i = t_{i-1} + \Omega(m_1 + m_1^{(i-1)} + m_2 + m_2^{(i-1)} + m_3 + C+m-i),
\end{equation}
there is an $(O(1), O(m_2 + m_2^{(i-1)}), \Ls_i, \Ts, t_i, \emptyset)$-controlled modification of the gadget that causes the NAND of the particles $p_{i,1}$ and $p_{i,2}$ to appear on the line $L_i$ at time $t_i$.

When all the particles are in their places, it remains to bound the number of new particles, occupied lines and crossings in the spacetime diagram of interest, and the final time $t' = t_{C + m}$. First, the number of particles and occupied lines increases by a constant each round, so we have $m_2^{(i)} = O(i)$. The number of crossings increases by $m_1^{(i)} - m_1^{(i - 1)} = O(m_2 + m_2^{(i-1)}) = O(m_2 + i)$ on round $i$, since each new occupied line crosses each already occupied line at most once, so we have $m_1^{(i)} = O(m_2 i + i^2)$. In particular, the total number of new crossings is $m_1^{(C + m)} = O((C + m_2 + m_4)(C + m_4))$. Any large enough time difference $t_i - t_{i-1} = \Omega(m_1 + m_2 i + i^2 + m_3 + C + m_4)$ suffices (where we have plugged the values of $m_1^i$ and $m_2^i$ into \eqref{eq:tdiff} and simplified a bit), so we can choose
\[ t_{C + m} - t_0 = \Omega((C + m_4)(C^2 + m_1 + m_3 + (C + m_4)(m_2 + m_4))), \]
which finishes the proof.
\qed
\end{proof}

\section{Physical Universality}
\label{sec:AllTogether}

We can now prove our main result, the effective physical universality of $f$.

\begin{theorem}
\label{thm:Main}
The cellular automaton $f$ is effectively physically universal.
\end{theorem}

\begin{proof}
We prove the result by a construction that uses the results from the previous sections. The reader should use Figure~\ref{fig:FinalProof} as a guide. Without loss of generality, let $D = E = [0, n-1]$, and let $h : A^n \to A^n$ be any function. We can view $h$ as a Boolean function $H : \{0, 1\}^{4n} \to \{0, 1\}^{4n}$ on arrangements of particles in the domain $D$. We set the initial configuration of interest to be $x \in \hat A^\Z$, with the fully general pattern containing the $4n$ Boolean particles with labels $\alpha_1, \ldots, \alpha_{4n} \in \mathcal{V}$ in the input, and 0 in the gadget. Our goal is to modify the gadget so that for some time $t_{\mathrm{final}} > 0$ polynomial in $n$ and the circuit complexity $C_H$ of $H$, we have $\hat f^{t_{\mathrm{final}}}(x)_i = H(\alpha_1, \ldots, \alpha_{4n})_{[4i, 4i+3]}$ for all $i \in [0, n-1]$. Lemma~\ref{lem:LogicRep} then gives the claim.

In the construction, we apply Lemma~\ref{lem:Diffusion} twice, first to the input pattern in the forward direction, to transform it to a dispersed pattern with $O(n)$ Boolean particles labeled $\beta_1, \ldots, \beta_m \in \F$, and then to the output pattern in the reverse direction, to find a dispersed pattern with $O(n)$ Boolean particles labeled $\delta_1, \ldots, \delta_{m'} \in \F$ at positions $(j_1, s_1), \ldots, (j_{m'}, s_{m'}) \in \Z \times \D$ that $\hat f$ transforms into the output. The function $H' : \{0, 1\}^m \to \{0, 1\}^k$ that we actually compute maps $\beta_1, \ldots, \beta_m$ to those elements of $\delta_1, \ldots, \delta_{m'}$ that are not identically zero, denoted $\delta_{n_1}, \ldots, \delta_{n_k}$, which is possible since the logical CA $\hat f$ is reversible. The function $H'$ has circuit complexity $C \leq C_H + O(n^2)$, since it can be computed by simulating in reverse the $O(n^2)$ collisions that produce $\beta_1, \ldots, \beta_m$ from the input, then the function $H$, and finally in reverse the $O(n^2)$ collisions that produce the output from $\delta_{n_1}, \ldots, \delta_{n_k}$.

The construction happens in five stages, scheduled at times $0 < t_{\mathrm{dis}} < t_{\mathrm{coll}} < t_{\mathrm{comp}} < t_{\mathrm{ass}} < t_{\mathrm{final}}$ with large enough differences. As stated above, we will produce the dispersed pattern $\delta_{n_1}, \ldots, \delta_{n_k}$ at time $t_{\mathrm{ass}}$, and $t_{\mathrm{final}} = t_{\mathrm{ass}} + O(n)$ is the time when this pattern has evolved into the desired output. Let $I = [-2(t_{\mathrm{final}} - t_{\mathrm{ass}}), n'/4 + 2(t_{\mathrm{final}} - t_{\mathrm{ass}})]$ be the interval containing the dispersed output pattern, and let $\Ls'_{\mathrm{ass}}$ be the size-$O(n)$ set of all lines that pass through a coordinate of $I$ at time $t_{\mathrm{ass}}$.

\paragraph{Diffusion Stage.} % no bounds needed
In the Diffusion Stage, we simply apply Lemma~\ref{lem:Diffusion} to the initial configuration of interest. It states that for some time $t_{\mathrm{dis}} = O(n)$, there are no crossings or collisions in the time interval $[t_{\mathrm{dis}}, \infty)$, and that the initial spacetime diagram of interest contains $O(n)$ occupied lines and $O(n^2)$ crossings. At time $t_{\mathrm{dis}}$, we have the aforementioned pattern of $O(n)$ Boolean particles, labeled $\beta_1, \ldots, \beta_{n'} \in \F$, moving away from the origin at different speeds.

\paragraph{Collection Stage.} % m_1 = n^2, m_2 = n, m_3 = C, m_4 = n
We now redirect the dispersed Boolean particles into a slow right-moving formation. For this, let $t_{\mathrm{coll}} = t_{\mathrm{dis}} + \Theta(C + n)$, and for $i \in \Z$, denote by $L_i =  L(i, t_{\mathrm{coll}}, 1)$ the line passing through position $(i,1)$ at time $t_{\mathrm{coll}}$. Let $\Ls_{\mathrm{coll}} = \{ L_{-i} \;|\; i = 1, \ldots, C + k \}$ be the $C + k$ lines reserved for computation, and let $\Ts_{\mathrm{coll}} = \{ L_i \;|\; i = 0, \ldots, m-1 \}$ be the target lines for the inputs of the computation. Since there are $O(n)$ occupied lines and $O(n^2)$ crossings in the spacetime diagram of interest, Corollary~\ref{cor:Movement} shows that there is an $(O(n), O(n^2), \Ls_{\mathrm{coll}} \cup \Ls'_{\mathrm{ass}}, \Ts_{\mathrm{coll}}, t_{\mathrm{coll}}, \emptyset)$-controlled modification that moves the Boolean particles onto the lines $L_1, \ldots, L_m$ at a large enough time $t_{\mathrm{coll}}$.

\paragraph{Computation Stage.} % m_1 = n^2, m_2 = n, m_3 = n, m_4 = n, C = C
In this stage, we compute the labels $\delta_{n_1}, \ldots, \delta_{n_k}$ by simulating the circuit representation of $H'$. For that, let $t_{\mathrm{comp}} = t_{\mathrm{coll}} + \Theta((C + n)^3)$ and $\Ts_{\mathrm{comp}} = \{ L_{-C-i} \;|\; i = 1, \ldots, k \}$. Corollary~\ref{cor:Computation} shows that there is an $(O(C + n), O((C + n)^2), \Ls'_{\mathrm{ass}}, \Ts_{\mathrm{comp}}, t_{\mathrm{comp}}, \emptyset)$-controlled modification which results in $k$ Boolean particles with labels $\delta_{n_1}, \ldots, \delta_{n_k}$ to be placed on the lines $L_{-C-1}, \ldots, L_{-C-k}$.

\paragraph{Assembly Stage.} % m_1 = (C + n)^2, m_2 = C + n, m_3 = n, m_4 = n
In the Assembly Stage, we assemble the dispersed output pattern in such a way that no auxiliary particle affects its evolution before the final time $t_{\mathrm{final}}$. For that, let $t_{\mathrm{ass}} = t_{\mathrm{comp}} + \Theta((C + n)^2)$, and let $\Ts_{\mathrm{ass}} = \{ L(j_{n_i}, t_{\mathrm{ass}}, s_{n_i}) \;|\; i = 1, \ldots, k \}$ be the target lines that pass through the positions of the Boolean particles in the dispersed output pattern at time $t_{\mathrm{ass}}$. From Lemma~\ref{lem:Diffusion} we know that the lines in $\Ts_{\mathrm{ass}}$ do not intersect in the time interval $(-\infty, t_{\mathrm{ass}}]$. We can then apply Corollary~\ref{cor:Movement} again to obtain a weakly $(O(n), O(n(C + n)), \emptyset, \Ts_{\mathrm{ass}}, t_{\mathrm{ass}}, I)$-controlled modification that places a Boolean particle with label $\delta_{n_i}$ on each spacetime position $(j_{n_i}, t_{\mathrm{ass}}, s_{n_i})$ while keeping the other spacetime positions in $I \times \{t_{\mathrm{ass}}\} \times \D$ empty.

\paragraph{Reverse Diffusion Stage.} % no bounds needed
Here we no longer need to do any modifications. Since we explicitly forbade any auxiliary Boolean particles from entering the spacetime segment $I \times \{ t_{\mathrm{ass}} \}$, the dispersed output pattern on it will evolve in $t_{\mathrm{final}} - t_{\mathrm{ass}} = O(n)$ steps into the final output pattern encoding $H(\alpha_1, \ldots, \alpha_{4n})$.
\qed
\end{proof}

\begin{figure}[ht]
\begin{center}
\begin{tikzpicture}[yscale=.9]

% n = 1

\draw[ultra thick] (0,0) -- (1,0);
\draw[ultra thick] (-2,1) -- (3,1);
\node (a) at (.5,0) {};
\foreach \l/\x in {b1/-1.5,b2/-.5,b3/1.5,b4/2.5}{
	\coordinate (\l) at (\x,.8) {};
	\draw[dotted,->] (a) -- (\l);
}

\draw[decorate,decoration={brace,amplitude=5}] (-2,1.2) -- (3,1.2);
\draw[decorate,decoration={brace,amplitude=5}] (1.5,2.8) -- (0,2.8);
\draw[densely dotted,->] (.5,1.5) -- (.5,1.7) -- (0,1.8) -- (1.25,2.2) -- (.75,2.3) -- (.75,2.5);

\foreach \x in {0,.2,...,1.5}{
	\fill (\x,3) circle (0.05);
	\draw (\x,3) -- ++(1,4);
}
\foreach \sa/\sb/\ta/\ti in {
		1/2/0/.25,
		0/2/1/.75,
		3/5/2/1,
		-1/4/3/1.5,
		-3/1/4/2,
		-3/-4/5/2.5,
		6/-4/6/2.75,
		-4/4/7/3.25}{
	\draw (\sa/5+\ti/4,\ti+3) -- (-\ta/5-1/5+\ti/4+1/8,\ti+7/2);
	\draw (\sb/5+\ti/4,\ti+3) -- (-\ta/5-1/5+\ti/4+1/8,\ti+7/2);
	\fill (-\ta/5-1/5+\ti/4+1/8,\ti+7/2) circle (0.05);
	\draw[dashed] (-\ta/5-1/5,3) -- (-\ta/5-1/5+\ti/4+1/8,\ti+7/2);
	\draw (-\ta/5-1/5+\ti/4+1/8,\ti+7/2) -- (-\ta/5-1/5+1,7);
}

\draw[decorate,decoration={brace,amplitude=5}] (-.6,7.2) -- (.2,7.2);
\draw[decorate,decoration={brace,amplitude=5}] (3,8.8) -- (-2,8.8);
\draw[densely dotted,->] (-.2,7.5) -- (-.2,7.7) -- (-.7,7.8) -- (1,8.2) -- (.5,8.3) -- (.5,8.5);

\draw[ultra thick] (0,10) -- (1,10);
\draw[ultra thick] (-2,9) -- (3,9);
\node (d) at (.5,10) {};
\foreach \l/\x in {c1/-1.5,c2/-.5,c3/1.5,c4/2.5}{
	\coordinate (\l) at (\x,9.2) {};
	\draw[>-,dotted] (\l) -- (d);
}

\draw[thick,decorate,decoration={brace,amplitude=5,mirror}] (4,0) -- (4,.9);
\node [right] at (4.2,.45) {Diffusion};
\draw[thick,decorate,decoration={brace,amplitude=5,mirror}] (4,1.1) -- (4,2.9);
\node [right] at (4.2,2) {Collection};
\draw[thick,decorate,decoration={brace,amplitude=5,mirror}] (4,3.1) -- (4,6.9);
\node [right] at (4.2,5) {Computation};
\draw[thick,decorate,decoration={brace,amplitude=5,mirror}] (4,7.1) -- (4,8.9);
\node [right] at (4.2,8) {Assembly};
\draw[thick,decorate,decoration={brace,amplitude=5,mirror}] (4,9.1) -- (4,10);
\node [right] at (4.2,9.7) {Reverse};
\node [right] at (4.2,9.3) {Diffusion};

\draw (-3,0) -- ++(0,10);
\draw[thick] (-3,0) -- ++(.25,0);
\node[left] at (-3,0) {$0$};
\draw[thick] (-3,1) -- ++(.25,0);
\node[left] at (-3,1) {$t_{\mathrm{dis}}$};
\draw[thick] (-3,3) -- ++(.25,0);
\node[left] at (-3,3) {$t_{\mathrm{coll}}$};
\draw[thick] (-3,7) -- ++(.25,0);
\node[left] at (-3,7) {$t_{\mathrm{comp}}$};
\draw[thick] (-3,9) -- ++(.25,0);
\node[left] at (-3,9) {$t_{\mathrm{ass}}$};
\draw[thick] (-3,10) -- ++(.25,0);
\node[left] at (-3,10) {$t_{\mathrm{final}}$};

\end{tikzpicture}
\end{center}
\caption{A schematic diagram of the proof of Theorem~\ref{thm:Main}, not drawn to scale.}
\label{fig:FinalProof}
\end{figure}
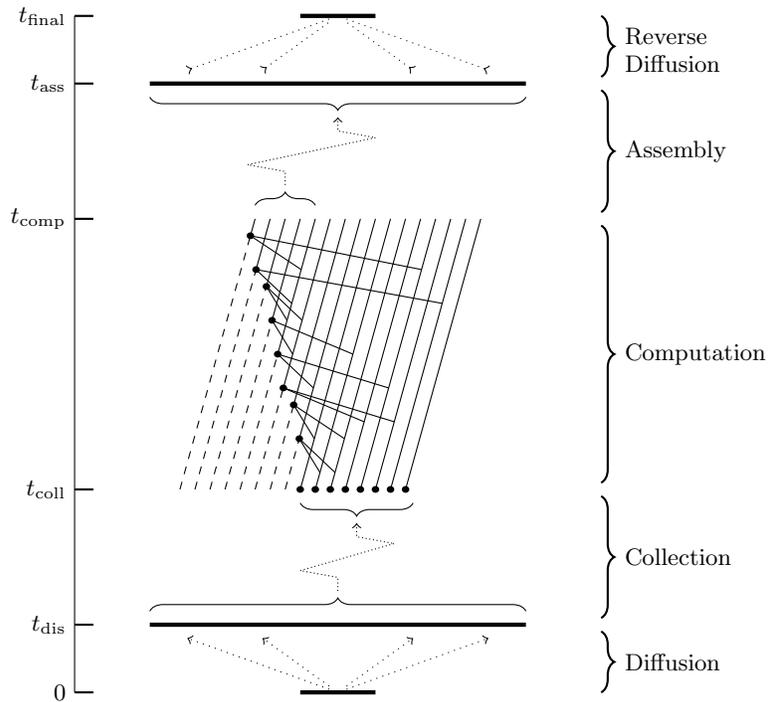

\section{Final remarks}
\label{sec:Final}

%Our construction of the physically universal in this paper has the drawback that it does not show the P-completeness of the automaton as a side-product. Namely, in our proofs we construct the gadget particle by particle, and when we need to add auxiliary particles to redirect existing ones, we essentially show by the pigeonhole principle that there exist legal positions for them.
Our proof of the physical universality of $f$ can readily be turned into a polynomial time algorithm that, given a circuit computing the function $h : A^D \to A^E$ in the definition of physical universality, computes the corresponding gadget and the polynomially bounded number $t_{\mathrm{final}}$. However, this algorithm will need polynomial space, as it compares the new positions of auxiliary particles to all existing ones. In fact, for technical reasons that could be easily avoided, namely due to our choice of handling constant-0 particles separately, constructing the gadget is P-complete. To construct the gadget in logarithmic space, it might be necessary to fix particular choices of where the auxiliary particles are put. We have chosen the more abstract route in the hope that our methods generalize more directly to a larger class of CA.

The existence of a physically universal CA was asked in \cite{Ja10} without fixing the number of states. Our CA has 16 states and radius 2. It would be interesting to find the minimal number of states and the minimal radii for physically universal CA. Of course, one can make any physically universal CA have radius 1 by passing to a blocking presentation, but this increases the number of states. From our CA, one obtains a physically universal radius-1 CA with 256 states.

\begin{question}
Are there physically universal CA on the binary alphabet? Which alphabet-radius pairs allow physical universality?
\end{question}

A long list of open questions about physical universality is also given in \cite{Sc14}.

%The simplest CA in terms of radius and number of states which are the number elementary CA, that is, CA with only 2 states and radius 1, discussed in \cite{Wo83}. Since a physically universal CA must be surjective,  \cite{KaSaTo13}.

\section*{Acknowledgments}

We are thankful to Charalampos Zinoviadis for introducing this problem to us, and for many fruitful discussions on the proof, and Luke Schaeffer for his Golly implementation of our physically universal CA. We would also like to thank Scott Aaronson for popularizing the concept in his blog \cite{Blog}.

\bibliographystyle{plain}
\bibliography{PhysBib}{}

\begin{thebibliography}{1}

\bibitem{Blog}
{S}htetl-{O}ptimized -- the blog of {S}cott {A}aronson.
\newblock \url{http://www.scottaaronson.com/blog/?p=1896}.
\newblock Accessed: 2014-09-17.

\bibitem{Wo02}
{\em A New Kind of Science}.
\newblock Wolfram Media Inc., Champaign, Ilinois, US, United States, 2002.

\bibitem{DeKuBl06}
Jean{-}Charles Delvenne, Petr Kurka, and Vincent~D. Blondel.
\newblock Decidability and universality in symbolic dynamical systems.
\newblock {\em Fundam. Inform.}, 74(4):463--490, 2006.

\bibitem{Ja10}
D.~{Janzing}.
\newblock {Is there a physically universal cellular automaton or Hamiltonian?}
\newblock {\em ArXiv e-prints}, September 2010.

\bibitem{Ka12}
Jarkko Kari.
\newblock Universal pattern generation by cellular automata.
\newblock {\em Theoretical Computer Science}, 429(0):180 -- 184, 2012.
\newblock Magic in Science.

\bibitem{Ne66}
John~Von Neumann.
\newblock {\em Theory of Self-Reproducing Automata}.
\newblock University of Illinois Press, Champaign, IL, USA, 1966.

\bibitem{Ol03}
Nicolas Ollinger.
\newblock The intrinsic universality problem of one-dimensional cellular
  automata.
\newblock In Helmut Alt and Michel Habib, editors, {\em STACS 2003}, volume
  2607 of {\em Lecture Notes in Computer Science}, pages 632--641. Springer
  Berlin Heidelberg, 2003.

\bibitem{Sc14}
Luke Schaeffer.
\newblock A physically universal cellular automaton.
\newblock {\em Electronic Colloquium on Computational Complexity {(ECCC)}},
  21:84, 2014.

\end{thebibliography}

\end{document}